\documentclass[11pt]{llncs}
\usepackage{amsmath,latexsym,amssymb}
\usepackage{pstricks,graphicx,rotating,hhline}
\usepackage{setspace}
\usepackage{multicol}

\newif\ifcode
\codefalse

\codetrue

\usepackage{macros}
\ifcode
\usepackage[ruled]{algorithm}
\usepackage{algpseudocode}
\usepackage{ioa_code}
\fi

\newtheorem{fact}[theorem]{Fact}

\renewenvironment{proof}{\noindent{\bf Proof.}}{\hfill$\Box$\FF}

\newcommand{\N}{\mathbb N}

\title{{\bf Timed Quorum Systems for\\ Large-Scale and Dynamic Environments}\thanks{\small Contact author: Vincent Gramoli, ASAP Research Group, INRIA Futurs, 3-4 rue Jacques Monod, 91893 Orsay, France; fax:  +33 1 74 85 42 42.}}

\author{Vincent Gramoli\inst{1}\fnmsep\inst{2}~~
Michel  Raynal\inst{2}}

\institute{Universit\'e de Rennes 1\\
INRIA Research Centre Rennes \\
Campus de Beaulieu,\\
35042 Rennes,  France}

\institute{
INRIA Futurs, \\Parc Club Orsay Universit\'e, 91893 Orsay, France \\
\email{vgramoli@irisa.fr}\\
\and Universit\'e de Rennes 1 and INRIA Research Centre Rennes,\\
Campus de Beaulieu, 35042 Rennes, France\\
\email{raynal@irisa.fr}\\
}

\begin{document}

\maketitle

\begin{abstract}
This paper presents Timed Quorum System (TQS), a new quorum 
system especially suited for large-scale and dynamic systems.
TQS requires that two quorums intersect with high probability if
they are used in the same small period of time.
It proposed an algorithm that implements TQS and that verifies
probabilistic atomicity: a consistency criterion that requires 
each operation to respect atomicity with high probability.
This TQS implementation has quorum of size $O(\sqrt{nD})$ and 
expected access time of $O(\log{\sqrt{nD}})$ message delays, where $n$ 
measures the size of the system and $D$ is a required parameter 
to handle dynamism. \\

\end{abstract}

\section{Introduction}


The need of resources is a main motivation behind distributed systems.
Take peer-to-peer (p2p) systems as an example. A p2p system is a distributed system that 
has no centralized control.
The p2p systems have gained in popularity with the massive utilization of file-sharing 
applications over the Internet, since 2000.  These systems propose a tremendous amount of 
file resources.
More generally, there is an increasing amount of various computing devices 
surrounding us:
IDC predicts that there will be 17 billions of traditional network devices
by 2012.
In such context, it is common knowledge that scalability has become one of 
the most important challenges of today's distributed systems.

The scale-shift of distributed systems modifies the way computational entities communicate.  
Energy dependence, disconnection, malfunctioning, and other environmental 
factors affect the availability of various computational entities independently.
This translates into unregular periods of activity during which an entity
can receive messages or compute tasks.
As a result of this independent and periodic behaviors, these systems are inherently highly dynamic.
Quorum system is a largely adopted solution for communication in message-passing system.
Despite the interest for emulating shared-memory in dynamic systems~\cite{Her87,RAMBO,FRT05,CGGMS05}, 
there is no scalable solution due to the cost of their failure handling mechanism 
or their operation complexity.
This paper proposes a new quorum system, Timed Quorum System (TQS), whose quorums 
have a bounded lifetime and that intersect with high probability during their lifetime.
We propose an implementation of TQS that emulates a probabilistic 
atomic memory, provided that each node is able to approximate the system size.
We show that the resulting quorum size is $O(\sqrt{nD})$. Factor $n$ is the number of nodes and
factor $D$ is required to handle the dynamism of nodes in the system and can
be bounded if operations are sufficiently frequent.
That is, quorum size becomes $O(\sqrt{n})$, which is optimal, as proved in~\cite{MRWW01}, 
for static settings.
Moreover, the expected time for an operation to contact a quorum is 
$O(\log{\sqrt{nD}})$ message delays.

\paragraph{Related Work.}

Dynamic quorum system is an active research topic.
Some dynamic quorums rely on failure detectors where if a failure is detected, 
then the quorum is adapted. This adaption leads to a redefinition of the 
quorums~\cite{Her87,NW98} or to the replacement of the failed nodes in the 
quorums~\cite{NN05,AM05,GAV07}. For example, in~\cite{AM05}, a communication
structure is continuously maintained to ensure that quorum intersects at all time
(with high probability).

Other solutions relies on periodic reconfigurations~\cite{RAMBO,CGGMS05} where 
the quorum systems are subsequently replaced.  These solutions are different from
the previous ones since the newly installed quorums do not need to intersect with 
the previous ones.  
In~\cite{FRT05} a quorum abstraction states requires two 
properties: (i)~intersection and (ii)~progress, in which the notion of time is 
introduced.  First, a quorum of a certain type intersects the quorum of another 
type contacted subsequently.  Second, each node of a quorum remains active 
between the time the quorum starts being probed and the time the quorum stopped 
being probed. 

As far as we know TQS is the first quorum system that expresses guarantees that 
are both timely and probabilistic.  Time and probability relax the 
traditional intersection requirement of quorums.
We present a scalable emulation of a probabilistic atomic memory where each 
operation is atomic with high probability and where expected operation message complexity is 
$O(\sqrt{nD})$ and expected operation time complexity is $O(\log{\sqrt{nD}})$.
If operations are sufficiently frequent then $D$ becomes a constant leading to 
quorum of size $O(\sqrt{n})$.

\section{System Model and Problem Definition}

\subsection{Model}

The computation model is very simple. 
The system consists of $n$ nodes. It is dynamic in the following sense. 
Every time unit, $cn$  nodes leave the system and $cn$ 
nodes enter the system, where $c$ is an upper bound on the percentage of 
nodes that enter/leave the system per time unit and is called the \emph{churn}; this  can be seen as  new 
nodes ``replacing'' leaving nodes. 
A node leaves the system either voluntarily or because it crashes.  
A node that leaves the system does not reenter it later.  (Practically, 
this means that, when a node reenters the system, it 
is considered as a new node; all its previous knowledge of the system 
state is lost.) 
For the sake of simplicity, it is assumed that for any subset $S$ of nodes, the 
portion of replaced nodes is $c|S|$.
As explained below, the model can be made more complex.
The \emph{universe} $U$ denotes all the nodes of the system, plus the ones that
have already leave the system and the ones that have not joined the system yet.

%

\subsection{Problem}
Most of the dynamic models assume that dynamic events are dependent from each other: only a 
limited number of nodes leave and join the system during a bounded period of time.
For instance in~\cite{CGGMS05}, it is assumed that nodes departures are dependent:
quorum replication ensures that all nodes of at least any two quorums remain active between two 
reconfigurations occur. 
However, in a real dynamic system, nodes act independently.
Due to this independence, even with a precise knowledge of the past dynamic events, 
one can not predict the future behavior of a node.
That is, putting this observation into the quorums context, it translates into the impossibility
of predicting deterministically whether quorums intersect.

In contrast, TQS requires that quorums intersect with high probability.
This allows to use a more realistic model in which there is a certain probability 
that nodes leave/join the system at the same time.
That is, the goal here is to measure the probability that quorum intersect while 
time elapses.  Observe that, realistically, the probability that $k$ nodes leave the system 
increases at the time elapses.  As a result, the probability that a 
quorum $Q(t)$ probed at time $t$ and that a quorum $Q(t')$ probed at time $t'$
intersect decreases as the period $|t'-t|$ increases.
In the following we propose an implementation of TQS where 
probability of intersection remains high.

More precisely, each quorum of our TQS implementation
is defined for a given time $t$.
Each quorum $Q(t)$ has a lifetime $\Delta$ that represents a period during which 
the quorum is \emph{reachable}.  Differently to availability defined in~\cite{NW98}, 
reachability does not depend on the number of nodes that are failed in a quorum system
because this number is unpredictable in dynamic systems.
Instead, a $Q(t)$ quorum is reachable if at least one node of quorum $Q(t)$
is reached with high probability: if two quorums are reachable at the same
time, they intersect with high probability.
More generally, let two quorums $Q(t)$ and $Q(t')$ of a TQS be reachable during 
$\Delta$ time (their lifetime is $\Delta$); if $|t-t'|\leq\Delta$ then  $Q(t)$ and 
$Q(t')$ intersect with high probability.

\paragraph{Probabilistic Atomic Object.}
Initially, any object has a default value $v_0$ that is replicated at a set of nodes and $V$ denotes the set of 
all possible values present in the system.  An object is accessed by read or write
operations initiated by some nodes $i$ at time $t\in T$ that returns or modify the object value $v$. ($T$ is the set 
of all possible time instants.)  If a 
node initiates an operation, then it is referred to as a \emph{client}.
All nodes of the system, including nodes of the quorum system, can initiate a read or a write 
operation, i.e., all nodes are potential clients and the multi-reader/multi-writer 
model is used.  In the following we only consider a single object accessed by operations that 
must satisfy probabilistic atomicity.

A probabilistic atomic object aims at emulating a memory
that offers high quality of service despite 
large scale and dynamism.
%
For the sake of tolerating scale-shift and dynamism, we aim at relaxing 
some properties.  However, our goal is to provide each client with a 
distributed shared memory emulation that offers satisfying quality
of service.  
Quality of service must be formally stated by a consistency criterion that 
defines the guarantees the application can expect from the memory emulation.  We 
aim at providing quality of service in terms of accuracy of
read and write operations.  In other words, our goal is to provide the clients with 
a memory that guarantees that each read or write operation will be successfully 
executed with high probability.
We define the probabilistic atomic object as an atomic object 
where operation accuracy is ensured with high probability.

Let us first recall properties 2 and 4 of atomicity from Theorem 13.16 of~\cite{Lyn96} which 
require that any sequence of invocations responses of read and write operations applied to $x$
satisfies a partial ordering $\prec$ such that:

\begin{itemize}
\item $(\pi_1,\pi_2)$-\textit{ordering}: if the response event of operation $\pi_1$ precedes the 
    invocation event of operation $\pi_2$, then it is not possible 
    to have $\pi_2 \prec \pi_1$;
\item $(\pi_1,\pi_2)$-\textit{return}: the value returned by a read operation $\pi_2$ is the value 
    written by the last preceding write operation $\pi_1$ regarding to $\prec$ 
    (in case no such write operation $\pi_1$ exists, this value returned is the 
    default value).
\end{itemize}
The definition of probabilistic atomicity is similar to the definition of atomicity: 
only Properties 2 and 4 are slightly modified, as indicated below.

\begin{definition}[Probabilistic Atomic Object]\label{def:patomicity}
Let $x$ be a read/write probabilistic atomic object.  Let $H$ be a complete sequence of invocations responses of read and write operations applied to object $x$. 
The sequence $H$ satisfies probabilistic atomicity if and only if there is a partial ordering $\prec$ on the successful operations such that the following properties hold:
\begin{enumerate}
\item For any operation $\pi_2$, there are only finitely many operations $\pi_1$, such that $\pi_1\prec \pi_2$.
  \item Let $\pi_1$ be a successful operation. Any operation $\pi_2$ satisfies $(\pi_1,\pi_2)$-\textit{ordering} with high probability.
(If $\pi_2$ does not satisfy it, then $\pi_2$ is considered as unsuccessful.)
  \item if $\pi_1$ is a write operation and  $\pi_2$ is any 
    operation, then either $\pi_2 \prec \pi_1$ or $\pi_1 \prec \pi_2$;
  \item Let $\pi_1$ be a successful operation. Any operation $\pi_2$ satisfies $(\pi_1,\pi_2)$-\textit{return} with high probability.
(If $\pi_2$ does not satisfy it, then $\pi_2$ is considered as unsuccessful.)
\end{enumerate}
\end{definition}

Observe that the partial ordering is defined on successful operations.  That is, 
either an operation $\pi$ fails and this operation is considered as unordered or the operation succeeds and is ordered 
with respect to other successful operations. 

Even though an operation succeeds with high probability, in an infinite execution it is very 
likely that at least one operation fails.  However, our goal is to provide the operation 
requester (client) with high guarantee of success for each of its operation request.

\paragraph{Additional Notations and Definitions.}

This paragraph defines several terms that are used in the algorithm description.
First, recall that a shared object is accessed through read operations, which return
the current value of the object, and write operations, which modify the current value 
of the object.  To clarify the notion of currency when concurrency happens, it is 
important to explain what are the up-to-date values that could be considered as current. 
We refer to the \emph{last value} as the value associated with the largest $\ms{tag}$
among all values whose propagation is complete.
We refer to the \emph{up-to-date values} at time $t$ as all values $v$ that satisfies 
one of the following properties: \textit{(i)}~value $v$ is the last value or
\textit{(ii)}~value $v$ is a value whose propagation is ongoing and whose associated tag
is at least equal or larger to the tag associated with the last value.

\section{Timed Quorum System}

This section defines Timed Quorum Systems (TQS).
Before being created of after its lifetime elapses, a quorum is not 
guaranteed to intersect with any other quorums, however, during 
its lifetime a quorum is considered as available: two quorums that
are available at the same time intersect with high probability.
In dynamic systems nodes may leave at any time, but this probability 
is bounded, thus it is possible to determine the intersection 
probability of two quorums.

\paragraph{Definition of Timed Quorum System (TQS).}

Next, we formally define TQS that are especially suited for dynamic systems.
Recall that the universe $U$ contains the set of all possible nodes, including the one that have not
join the system yet.
First, we restate the definition of a \emph{set system} as a set of subsets of a universe of nodes.

\begin{definition}[Set System]\label{def:setsystem}
  A \emph{set system} ${\cal S}$ over a universe $U$ is a set of subsets of $U$.
\end{definition}

Then, we define the timed access strategy as an access strategy over a set system that may vary over time.
This definition is motivated by the fact that an access strategy defined over a set ${\cal S}$ can evolve.
To compare with the existing probabilistic dynamic quorums, in~\cite{AM05} the authors defined a dynamic quorum system using an evolving
strategy that might replace some nodes of a quorum while its access strategy remains identical 
despite this evolution.  Unlike the dynamic quorum approach, we need a more general framework to 
consider quorums that are different not only because of their structure but also because of how likely 
they can be accessed.
The timely access strategy 
adds a time parameter to the seminal definition access strategy given by Malkhi et al.~\cite{MRWW01},
A timely access strategy is allowed to evolve over time.

\begin{definition}[Timed Access Strategy]
A \emph{timed access strategy} $\omega(t)$ for a set system ${\cal S}$ at time $t\in T$ is a
probability distribution on the elements of ${\cal S}$ at time $t$.  That is, 
$\omega : {\cal S} \times T \rightarrow [0,1]$ satisfies at any time $t\in T$:
$\sum_{s\in{\cal S}}\omega(s,t) = 1$.
\end{definition}

Informally, at two distinct instants $t_1 \in T$ and $t_2 \in T$, an access strategy might be different
for any reason. For instance, consider that some node $i$ is active at time $t_1$ while the same node $i$ is failed at time 
$t_2$, hence it is likely that if $i\in s$, then $\omega(s,t_1) \neq 0$ while 
$\omega(s,t_2) = 0$.  This is due to the fact that a node is reachable only when it is active.

\begin{definition}[$\Delta$-Timed Quorum System]\label{def:timedqs}
  Let ${\cal Q}$ be a set system, let $\omega(t)$ be a timed access strategy for ${\cal Q}$ at time $t$, 
  and let $0 < \epsilon < 1$ be given.
  
  The tuple $\tup{{\cal Q}, \omega(t)}$ is a \emph{$\Delta$-timed quorum system} if for any
quorums $Q(t_1) \in {\cal Q}$ accessed with strategy $\omega(t_1)$ and $Q(t_2) \in {\cal Q}$ 
accessed with strategy $\omega(t_2)$, 
we have:
\begin{eqnarray}
\Delta \geq |t_1 - t_2| &\Rightarrow& \Pr[Q(t_1) \cap Q(t_2) \neq \emptyset] \geq 1 - \epsilon. \notag
\end{eqnarray}
\end{definition}

\section{Timed Quorum System Implementation for Probabilistic Atomic Memory}

In the following, we present a completely structureless memory.
The quorum systems this memory uses does not rely on any structure which 
makes it flexible.  In contrast with using a logical structured overlay (e.g.,~\cite{MKKB01}) for
communication among quorum system nodes, we use an unstructured communication overlay~\cite{GKM03}.
The lack of structure presents several benefits.  First, there is no need
to readapt the structure at each dynamic event. Second, there is no need for 
detecting failure.  
%
Our solution proposes a periodic replication.
To ensure 
the persistence of an object value despite unbounded
leaves, the value must be replicated an unbounded number of times.  
The solution we propose requires periodic operations and an 
approximation of the system size.  Although we do not focus 
on the problem of approximating the system size $n$, we suggest
the use of existing protocols approximating closely the system size
in dynamic systems~\cite{LKM06}.


\paragraph{Replicating during client operations.}

Benefiting from the natural primitive of the distributed
shared memory, values are replicated using operations.
Any operation has at its heart a quorum-probe that replicates value.
On the one hand, it is natural to think of a write operation as an 
operation that replicates a value. On the other hand, in~\cite{AW98} a 
Theorem shows that "read must write", meaning that a read operation 
must replicates the value it returns.
This raises the question: if operations replicate, why does a memory 
need additional replication mechanism?
%
In large-scale systems, it is also reasonable to assume that shared
objects are frequently accessed because of the large number of participants.
%
%

\paragraph{Quorum Probe.}

The algorithm is divided in three distinct parts that represent the state of the algorithm (Lines~\ref{line:state-start}--\ref{line:state-end}), the actions initiated by a client (Lines~\ref{line:active-start}--\ref{line:active-end}), and the actions taken upon reception 
of messages by a node (Lines~\ref{line:passive-start}--\ref{line:passive-end}), respectively.
Each node $i$ has its own copy of the object called its value $\ms{val}_i$
and an associated tag $\ms{tag}_i$. Field $\ms{tag}$ is a couple of a counter 
and a node identifier and represents, at any time, the version number of 
its corresponding value $\ms{val}$.  We assume that, initially, there are $q$ nodes that own the default
value of the object, the other nodes have their values $\ms{val}$ set to $\bot$ and 
all their $\ms{tag}$s are set to $\tup{0,0}$.

\begin{algorithm}[!ht]
\caption{Disseminating Memory at node $i$}\label{alg:gossip}
  \begin{algorithmic}[1]{\size
  
   	\Part{State of node $i$}{																										\label{line:state-start}
   		\State $q = \frac{\beta\sqrt{n}}{(1-c)^\frac{\Delta}{2}}$, the quorum size						\label{line:qsize}
	    \State $\ell, k\in \N$ the disseminating parameters taken such that $\frac{k^{l+1}-1}{k-1} \geq q$\label{line:dissemiate-size}
	    \State $\ms{val} \in V$, the value of the object, initially $\bot$
	    \State $\ms{tag}$, a couple of fields:
	    \State \T$\ms{counter}\in \N$, initially $0$
	    \State \T$\ms{id}\in I$, an identifier initially $i$
	    \State $\ms{marked}$, an array of boolean initially $\lit{false}$ at all indices 
	    \State $\ms{sent-to-nbrs1}$, $\ms{sent-to-nbrs2}$ two sets of node identifiers, initially $\emptyset$ 
	    \State $\ms{rcvd-from-qnodes}$, an infinite array of identifier sets, initially $\emptyset$ at all indices 
	    \State $\ms{sn} \in \N$, the sequence number of the current phase, initially 0
	  }\EndPart																																		\label{line:state-end}
  
%

		\Statex
		
		\begin{multicols}{2}
		
		\Part{Read$_i$}{																											\label{line:active-start}
			\State $\tup{\ms{val},\ms{tag}} \gets$ {\bf Consult}()							
			\State {\bf Propagate}($\tup{\ms{val},\ms{tag}}$)										\label{line:dissem-tag-r}
		}\EndPart
		
		\newpage

		\Part{Write($\ms{v}$)$_i$}{	 
			\State $\tup{*,\ms{tag}} \gets ${\bf Consult}()
			\State $\ms{tag}.\ms{counter} \gets \ms{tag}.\ms{counter}+1$
			\State $\ms{tag}.\ms{id} \gets i$
			\State $\ms{val} \gets v$
			\State {\bf Propagate}($\tup{\ms{val},\ms{tag}}$)										\label{line:dissem-tag-w}
		}\EndPart
		
		\end{multicols}
		
		\Part{Consult$_i$}{	 
	    \State $\ms{ttl} \gets \ell$
	    \State $\ms{sn} \gets \ms{sn}+1$
	    \While{$(|\ms{sent-to-nbrs1}| < k)$}																										\label{line:snd-nbrs-start1}
	      \State $\act{send} \tup{\lit{CONS}, \ms{val}, \ms{tag}, \ms{ttl}, i, \ms{sn}}$ to $(k- |\ms{sent-from-nbrs1}|)$ 
	      neighbors $\neq j$
	      \State $\ms{sent-to-nbrs1} \gets\ms{sent-to-nbrs1}\cup \{j\}$
      \EndWhile  																																										\label{line:snd-nbrs-end1}
      \State $\ms{sent-to-nbrs1} \gets \emptyset$
	    \WUntil{$|\ms{rcvd-from-qnodes}[\ms{sn}]| \geq q$}\EndWUntil																\label{line:snd-quorum-end1}
	    \State {\bf return}~$(\tup{\ms{val}, \ms{tag}})$
		}\EndPart
		

    
       }
    \algstore{disseminate}
   \end{algorithmic}
\end{algorithm}

Each read and write operation is executed by client $i$ in two subsequent phases, each 
disseminating a message to $q = O(\sqrt{nD})$ nodes, where $D=1/(1-c)^{\Delta}$ is required
to handle churn $c$ during period $\Delta$.\footnote{In~\cite{MRWW01}, it has 
been showed that $q = O(\sqrt{n})$ is sufficient in static systems.}
The two subsequent phases are called the \emph{consultation phase} and the \emph{propagation phase}.  The consultation 
phase aims at consulting the up-to-date value of the object that is present
in the system. (This value is identifiable since it associates the
largest tag present in the system.) 
More precisely, client $i$ disseminates a consultation message to $q$ nodes so that
each receiver $j$ responds with a message containing value $\ms{val}_j$ 
and tag $\ms{tag}_j$ so that client $i$ can update $\ms{val}_i$ and $\ms{tag}_i$.  
In fact, $i$ updates $\ms{val}_i$ and $\ms{tag}_i$ if and only if 
the $\ms{tag}_i$ has either a smaller counter than $\ms{tag}_j$ or it has an equal
counter but a smaller identifiers $i<j$ (node identifiers are always distinct); in this case we say
$\ms{tag}_i < \ms{tag}_j$ for short (cf. Lines~\ref{line:comparison1} and \ref{line:comparison2}).
Ideally, at the end of the consultation phase
client $i$ has set its value $\ms{val}_i$ to the up-to-date value.
Read and write operations differ from the value and tag that are propagated by the client 
$i$.  Specifically, in case of a read, client $i$ propagates the value and tag pair freshly consulted, 
while in the case of write, client $i$ propagates the new value to write
with a strictly larger tag than the largest tag that $i$ has consulted so far.
The propagation phase propagates the corresponding value and tag by dissemination among nodes.

\begin{algorithm}[!ht]
  \begin{algorithmic}[1]{\size
  \algrestore{disseminate}
  	
  	\Part{Propagate(\tup{\ms{val},\ms{t}})$_i$}{	 
	    \State $\ms{ttl} \gets \ell$
	    \State $\ms{sn} \gets \ms{sn}+1$
	    \While{$(|\ms{sent-to-nbrs1}| < k)$}																							\label{line:snd-nbrs-start2}
	      \State $\act{send} \tup{\lit{PROP}, \ms{val}, \ms{tag}, \ms{ttl}, i, \ms{sn}}$ to $(k- |\ms{sent-to-nbrs1}|)$ 
	      neighbors $\neq j$
	      \State $\ms{sent-to-nbrs1} \gets\ms{sent-to-nbrs1}\cup \{j\}$
      \EndWhile 																																												  \label{line:snd-nbrs-end2}
      \State $\ms{sent-to-nbrs1} \gets \emptyset$
	    \WUntil{$|\ms{rcvd-from-qnodes}[\ms{sn}]| \geq q$}\EndWUntil																			\label{line:snd-quorum-end2}
		}\EndPart																																															\label{line:active-end}

  	\Statex
  
  	\Part{Participate$_i$ (Activated upon reception of a message)}{													\label{line:passive-start}
	    \State $\act{recv} \tup{\ms{type}, \ms{v}, \ms{t}, \ms{ttl}, \ms{client-id}, \ms{sn}}$ from $j$
		  
		  \If{$(\ms{marked}[\ms{sn}])$}																																			\label{line:forward-start}
		      \State $\act{send} \tup{\ms{type}, \ms{v}, \ms{t}, \ms{ttl}, \ms{client-id}, \ms{sn}}$ to 
		      a neighbor $\neq j$
	    \Else 																					\label{line:forward-end-prob}
	    	\State $\ms{marked}[\ms{sn}] \gets \lit{true}$																							\label{line:participate-start}
			  \SmallIf{$(\ms{type} = \lit{CONS})$} $\tup{v,t} \gets \tup{\ms{val},\ms{tag}}$ \EndSmallIf
				\SmallIf{$(\ms{type} = \lit{PROP})$} $\tup{\ms{val},\ms{tag}} \gets \tup{v,t}$ \EndSmallIf	\label{line:comparison1}
			  \If{$(\ms{type} = \lit{RESP})$}
			  	\SmallIf{$\ms{tag} < t$} $\tup{\ms{val},\ms{tag}} \gets \tup{v,t}$ \EndSmallIf						\label{line:comparison2}
		    	\State $\ms{rcvd-from-qnodes}[\ms{sn}] \gets\ms{rcvd-from-qnodes}[\ms{sn}] \cup \{j\}$
			  \EndIf
			  
			  \State $\ms{ttl} \gets \ms{ttl} - 1$
		    \If{$(\ms{ttl} > 0)$} 
		    
		    	\While{$(|\ms{sent-to-nbrs2}| < k)$}	
			      \State $\act{send} \tup{\ms{type}, \ms{v}, \ms{t}, \ms{ttl}, \ms{client-id}, \ms{sn}}$ to $(k- |\ms{sent-to-nbrs2}|)$ 
			      neighbors $\neq j$
			      \State $\ms{sent-to-nbrs2} \gets\ms{sent-to-nbrs2}\cup \{j\}$
		      \EndWhile
		      \State $\ms{sent-to-nbrs2} \gets\emptyset$
		    \EndIf
				\State send $\tup{\lit{RESP}, \ms{val}, \ms{tag}, \ms{ttl}, \bot, \ms{sn}}$ to $\ms{client-id}$\label{line:participate-end}
		  \EndIf																																					\label{line:passive-end}
%
%
%
		}\EndPart																																								
  	
   }
   \end{algorithmic}
\end{algorithm}

Next, we focus on the dissemination procedure that is at the heart of the 
consultation and propagation phases. 
There are two parameters, $\ell, k$, that define the way all consultation or propagation messages
are disseminated.  Parameter $\ell$ indicates the depth of the dissemination, it is used to set a 
time-to-live field $\ms{ttl}$ that is decremented at each intermediary node that participates in 
the dissemination; if $\ms{ttl}=0$, then dissemination is complete.
Parameter $k$ represents the number of neighbors that are contacted by each intermediary participating 
node.  Together, parameters $\ell$ and $k$ define the number of nodes that are contacted during a 
dissemination. This number is $\frac{k^{\ell+1}-1}{k-1}$ (Line~\ref{line:dissemiate-size}) and represents 
the number of nodes in a balanced 
tree of depth $\ell$ and degree $k+1$: each node having at most $k$ children.  
(This value is provable by recurrence on the depth $\ell$ of the tree.)  
Observe that $\ell$ and $k$ are chosen such that the number of nodes that are contacted 
during a dissemination be larger than $q$ as written Line~\ref{line:dissemiate-size}.

There are three kind of messages denoted by message $\ms{type}$: $\lit{CONS}$, $\lit{PROP}$,
$\lit{RESP}$ indicating if the message is a consultation message, a propagation message, 
or a response to any of the two other messages.
When a new phase starts at client $i$, a time-to-live field $\ms{ttl}$ is set to $\ell$ 
and a sequence number $\ms{sn}$ is incremented.
This number is used in message exchanges to indicate whether a message corresponds to 
the right phase.
Then the phase proceeds in sending continuously messages to $k$ neighbors waiting for their 
answer (Lines~\ref{line:snd-nbrs-start1}--\ref{line:snd-nbrs-end1} and 
Lines~\ref{line:snd-nbrs-start2}--\ref{line:snd-nbrs-end2}).
When the $k$ neighbors answer, client $i$ knows that the dissemination is ongoing.
Then client $i$ receives all messages until a large enough number $q$ of nodes 
have responded in this phase, i.e., with the right sequence number (Lines~\ref{line:snd-quorum-end1}, \ref{line:snd-quorum-end2}).  
If so, then the phase is complete.

Observe that during the dissemination, messages are simply marked (if not so), responded (to client $i$), and reforwarded 
to other neighbors (until $\ms{ttl}$ is null).  Messages are marked by the node $i$ that participates into 
a dissemination for preventing node $i$ from participating multiple times in the same dissemination (Line:~\ref{line:forward-start}). 
As a result, if node $i$ is asked several times to participate, it first participates (Lines~\ref{line:participate-start}--\ref{line:participate-end}) and then it asks another node to participate 
(Lines~\ref{line:forward-start}--\ref{line:forward-end-prob}).  
More precisely, if $\ms{marked}[\ms{sn}]$ is true, then node $i$ re-forwards messages of sequence number $\ms{sn}$ 
without decrementing the $\ms{ttl}$. 
Observe that phase termination and dissemination termination depends on the number of participants rather than the 
number of responses:
it is important that enough participants participate in each dissemination for the phase to eventually 
end.

\paragraph{Contacting Participants Randomly.}
In order to contact the participants randomly, we implemented a membership
protocol~\cite{GKM03}.  This protocol is based on Cyclon~\cite{VGS05}, thus, it 
is lightweight and fault-tolerant. 
Each node has a set of $m$ neighbors
called its view ${\cal N}_i$, it periodically updates its view and recomputes 
its set of neighbors.  
Our underlying membership algorithm provides each node with
a set of $m \geq k+1$ neighbors, so that phases of Algorithm~\ref{alg:gossip} 
disseminate through a tree of degree $k+1$.
%
This algorithm shuffles the view at each cycle
of its execution so that it provides randomness in the choice of neighbors.
Moreover, it has been shown by simulation that the communication graph obtained
with Cyclon is similar to a random graph where neighbors are picked uniformly 
among nodes~\cite{I05}.
Finally, for a different purpose we already have simulated this variant of Cyclon
in~\cite{FGJKR07}: the results obtained was really similar to the one obtained with artificial 
uniformity.


For the sake of uniformity, the membership procedure
is similar to the Cyclon algorithm: 
each node $i$ maintains a view ${\mathcal N_i}$ containing one
entry per neighbor.  
The entry of a neighbor $j$ corresponds to a tuple containing
the neighbor identifier and its age.
Node $i$ copies its view, selects the
oldest neighbor $j$ of its view, removes the entry $e_{j}$ of $j$ from the copy
of its
view, and finally sends the resulting copy to $j$.  When $j$ receives the
view, $j$ sends its own view back to $i$ discarding possible pointers to 
$i$, and $i$ and $j$ update their view with the one they 
receive by firstly keeping the entries they received.  The age of neighbor $j$ entry 
denotes the time that elapsed since the last message from $j$ has been received; this is
used to remove failed neighbor from the list. This variant of Cyclon
exchanges all entries of the view at each step and uses two additional parameters.

\section{Correctness and Performance Results}

This Section gives the result of our algorithm. 
We assume that, initially, at least $q$ nodes own the default value
of the object. Assume also that at least one propagation phase from a successful operation starts 
every $\Delta$ time units and let the time of any phase be bounded by $\delta$ time units.
Next, we assume that 
our underlying communication protocol
provides each node with a view that represents a set of neighbors uniformly drawn at random among the 
set of all active nodes.  Recall that Cyclon shuffles node views and provides communication graph 
similar to a random graph~\cite{I05}.

The first Theorem shows that the proposed solution implements a TQS. 
The second Theorem shows that our solution satisfies probabilistic atomicity.
By lack of space, the proofs are given in the Appendix.

\begin{theorem}
Algorithm~\ref{alg:gossip} implements a $\Delta$-Timed Quorum System, where $\Delta$ is the maximum time between 
two subsequent propagation starts.
\end{theorem}

\begin{theorem}
Algorithm~\ref{alg:gossip} implements a probabilistic atomic object. 
\end{theorem}

Next Lemmas show the performance of our solution: 
the first Lemma gives the message complexity of our solution while the second Lemma gives
the time complexity of our solution.
Observe first that operations complete
provided that sent messages are reliably delivered.  Building onto this assumption, 
an operation complete after contacting $O(\sqrt{nD})$ nodes.  The following Lemma
shows this result.

\begin{lemma}
If messages are not lost, an operation complete after having contacted $O(\sqrt{nD})$ nodes.
\end{lemma}
\begin{proof}
This is straightforward from the fact that termination of the dissemination process
is conditioned to the number of distinct nodes contacted: $q = O(\sqrt{nD})$, with $D = (1-c)^{-\Delta}$ 
(cf. Line~\ref{line:qsize}). 
Since there are two disseminating phases in each operation, an operation is executed after 
contacting $O(\sqrt{nD})$ nodes. 
\end{proof}

Next Lemma indicates that an operation terminates in $O(\log{\sqrt{nD}})$ message delays, in expectation.

\begin{lemma}
If messages are not lost, the expected time of an operation is $O(\log{\sqrt{nD}})$ message delays.
\end{lemma}
\begin{proof}
The proof relies on the fact that $q'$ nodes are contacted uniformly at
random with replacement. In expectation, the number $q'$ that must be contacted to obtain $q$ 
distinct nodes is $q' = q = O(\sqrt{nD})$. Since nodes are contacted in parallel along a tree
of depth $\ell$ and degree $k+1$, the time required to contact all the nodes on the tree 
is $\ell = O(\log_k{q})$.  That is, it is done in $\ell = O(\log_k{\sqrt{nD}})$ 
message delays.
\end{proof}

%

\section{Conclusion}

This paper addressed the problem of emulating a distributed shared memory that 
tolerates scalability and dynamism while being efficient.
%
TQS ensures probabilistic intersection of quorums in a timely fashion.
%
Interestingly, we showed that some TQS implementation verifies a consistency 
criterion weaker but similar to atomicity: probabilistic atomicity.  Hence, 
any operation provided by some TQS satisfies the ordering required
for atomicity with high probability.  
The given implementation of TQS verifies probabilistic atomicity, provides 
lightweight ($O(\sqrt{nD})$ messages) and fast ($O(\log{\sqrt{nD}})$ message delays) 
operations, and does not require reconfiguration mechanism since
periodic replication is piggybacked into operations.

Since we started tackling the problem that node can fail independently,
we are now able to implement probabilistic memory into more realistic models.
Previous solutions required that a very few amount of nodes could fail
at the same time.  More realistically, a model should allow node to act independently
while requiring that failures occurring at the same time are unlikely.
An interesting question is: what probabilistic consistency can TQS achieve
in such a realistic model? \\

%

\noindent {\bf Acknowledgments.}
We are grateful to Anne-Marie Kermarrec and Achour Most\'efaoui for
fruitful discussions about gossip-based algorithms and dynamic systems.

\bibliographystyle{splncs}
\bibliography{thesis}

\appendix

\section{Correctness Proof}\label{app:proof}

Here, we show that Algorithm~\ref{alg:gossip} implements 
a timed quorum system and that it emulates the probabilistic atomic 
object abstraction defined in Definition~\ref{def:patomicity}.
The key points of this proof is to show that quorums are sufficiently 
re-activated by new operations to face dynamism and that subsequent quorums
intersect with very high probability to achieve probabilistic atomicity.

\paragraph{Assumptions and notations.}
First, we only consider executions starting with at least $q$ nodes that own the default value
of the object.  In these executions, at least one propagation phase from a successful operation starts 
every $\Delta$ time units and let the time of any phase be bounded by $\delta$ time units.
We assume that during a propagation that propagates a value $v$ to $q$ nodes and that
executes between time $t$ and $t+\delta$, 
there is at least one instant $t'$ where the $q$ nodes own value $v$ simultaneously.
This instant, $t'$, can occur arbitrarily between time $t$ and
$t+\delta$.  Even if this assumption may not seem realistic since
propagation occurs in parallel of churn (i.e., at the time the propagation contacts the 
$q^{th}$ node the first contacted node may have left the system), our motivations for this assumption
comes from the sake of clarity of the proof and we claim that the absence of 
this assumption leads to the same results. 

Second, we assume that 
our underlying communication protocol
provides each node with a view that represents a set of neighbors uniformly drawn at random among the 
set of all active nodes.  
This assumption is reasonable since, as already mentioned, the underlying algorithm is 
based on Cyclon that shuffles node views and provides communication graph similar to a 
random graph~\cite{I05}.

Next, we show that Algorithm~\ref{alg:gossip} implements a probabilistic object.
Observe that the liveness part of this proof relies simply on the activity of neighbors,
and the fact that messages are eventually received.  More precisely, by examination of 
the code of Algorithm~\ref{alg:gossip}, 
messages are gossiped among neighbors while neighbors are uniformly chosen.  It is clear that operation
termination depends on eventual message delivery.  As a result, only the safety part of 
the proof follows. 
In the following, $\ms{val}(\phi)$ (resp. $\ms{tag}(\phi)$) denote,
the value (resp. tag) consulted/propagated by phase $\phi$.

\paragraph{Correctness proof.}
First, we restate a Lemma appeared in~\cite{GKMRS06b} that computes the ratio of nodes that 
leave the system as time elapses, given a churn of $c$.
The result is the ratio of nodes that leave and join, and helps
computing the probability that up-to-date values remain reachable despite dynamism.

\begin{lemma}\label{lem:churn}
The ratio of initial nodes that have been replaced after $\tau$ time units
is at most $C=1-(1-c)^\tau$.
\end{lemma}

For the proof of the above Lemma~\ref{lem:churn}, please refer to~\cite{GKMRS06b}.
%
The following Lemma gives a lower bound on the number of nodes that own the up-to-date
value at any time in the system. (Recall that an up-to-date value is either 
the value with the largest tag and whose propagation is complete, or 
any value with a larger tag, but whose propagation is ongoing.) 

\begin{lemma}\label{lem:withprop}
At any time $t$ in the system, the number of nodes that own an up-to-date value is at least 
$q(1-c)^{\Delta}$, where
$\Delta$ is the maximum time between two subsequent propagation starts, $q$ is the quorum size, 
and $c$ is the churn of the system.
\end{lemma}
\begin{proof} 
  With no loss of generality, let $\rho_1, ..., \rho_k$ be all the ongoing propagations 
  at time $t$ and let $\rho_0$ be the latest successful propagation
  that is already finished at time $t$.   
  By definition, all $v(\rho_i)$ for any $i\geq 0$ are the up-to-date values in the system. 
Propagations $\rho_1, ..., \rho_k$ must all have started after time $t-\delta$. 
By the periodicity assumption of propagation phase,
propagation $\rho_0$ can not start earlier than time 
$t-\Delta+\delta$. 
Due to propagation $\rho_0$, 
there must be $q$ nodes with value $v(\rho_0)$ between times
$t-\Delta+\delta$ and $t-\Delta+2\delta$.

Since the number of replaced nodes increases as time elapses, assume a worst case
scenario in which $q$ nodes own value $v(\rho_0)$ at time $t_1 = t-\Delta+\delta$, we show
that at least $q(1-c)^{\Delta}$ nodes with 
value $v(\rho_0)$ remain in the system at time $t_2 = t+\delta$. 
By Lemma~\ref{lem:churn}, we know that during period $t_2 - t_1 = \Delta$ exactly 
$\lfloor q(1-(1-c)^\Delta) \rfloor$ nodes with value $v(\rho_0)$ are replaced. 
Since propagations $\rho_1, ..., \rho_k$ are ongoing, there may be some
successful propagations among those ones that 
overwrite some node values.  
Observe that if this overwriting happens only to nodes that already own value $v(\rho_i)$, 
then the number of nodes with value $v(\rho_i)$ remains at least
$q(1-c)^\Delta$ at time $t+\delta$; if this overwriting happens
to nodes that do not own value $v(\rho_i)$ then this number increases.
That is, $q(1-c)^\Delta$ is a lower bound of the number
of nodes with value $v(\rho_i)$ at time $t+\delta$, which leads to the result.
\end{proof}

The following Fact gives this well-known bound on the exponential function, provable using the Euler's method.

\begin{fact}\label{fact}
$(1+\frac{x}{n})^n \leq e^{x}$ for $n>|x|$. 
\end{fact}

Next Lemma lower bounds the probability that any consultation consults an up-to-date value $v$. 
Recall that sometime it might happen that a value $v'$ is unsuccessfully propagated.
This may happen when a write operation fails in consulting the largest tag just before 
propagating value $v'$. 
Observe that in any case, a successful consultation returns only successfully propagated values.

\begin{lemma}
\label{lem:proba2}
If the number of nodes that own an up-to-date value is at least $q(1-c)^{\Delta}$ during 
the whole period of execution of consultation $\phi$, then consultation $\phi$ succeeds with 
high probability ($\geq 1 - e^{-\beta^2}$, with $\beta$ a constant).
\end{lemma}
\begin{proof}
The consultation of Algorithm~\ref{alg:gossip} draws uniformly at 
random $q$ nodes, without replacement.
To lower bound the probability ${\cal P}$ that any consultation consults an up-to-date value $v$,
we compute the probability that this value is obtained after $q$ drawings with 
replacement.  It is clear that the probability of obtaining a specific
node after $q$ drawings is larger without replacement than with replacement.  
  The probability for a node $x$ uniformly chosen at random 
  not to own the value $v$ is
$\Pr[x\notin {\cal Q}] = 1-\frac{q(1-c)^\Delta}{n}$
that is, the probability not to consult value $v$ after $q$ drawings, with 
replacement, is
$\Pr[x_1\notin {\cal Q}, ..., x_{q}\notin {\cal Q}] = \left(1-\frac{q(1-c)^\Delta}{n}\right)^{q}$.
By Fact~\ref{fact}, $\Pr[x_1\notin {\cal Q}, ..., x_{q}\notin {\cal Q}] \leq e^{-\frac{q^2}{n}(1-c)^\Delta}.$
By replacing the $q$ by the quorum size given at Line~\ref{line:qsize} 
of Algorithm~\ref{alg:gossip} in 
the contrapositive ${\cal P} \geq 1 - e^{-\frac{q^2}{n}(1-c)^\Delta}$ we obtain
the result ${\cal P} \geq 1 - e^{-\beta^2}$.
\end{proof}

This corollary simply concludes the two previous Lemmas stating that any consultation 
executed in the system succeeds by returning an up-to-date value.

\begin{corollary}\label{cor:prob-cons}
Any consultation $\phi$ succeeds with high probability ($\geq 1 - e^{-\beta^2}$, with $\beta$ a constant).
\end{corollary}
\begin{proof}
  The result is straightforward from Lemma~\ref{lem:withprop} and Lemma~\ref{lem:proba2}.
\end{proof}

Last but not least, the two theorems conclude the proof by showing that Algorithm~\ref{alg:gossip} implements a $\Delta$-TQS and
verifies probabilistic atomicity.

\setcounter{theorem}{0}
\begin{theorem}
Algorithm~\ref{alg:gossip} implements a $\Delta$-Timed Quorum System, where $\Delta$ is the maximum time between 
two subsequent propagation starts.
\end{theorem}
\begin{proof}
First observe that the set of quorums is the set of subsets of $q$ active nodes over the system at time $t$.
The timed access strategy at time $t$ over the set of all quorums is the uniform access strategy over all quorums since each 
node is chosen with a uniform access strategy among the active nodes at time $t$.
By Corollary~\ref{cor:prob-cons}, it is clear that the intersection between two quorums is ensured with high probability
as long as one quorum starts being contacted $\Delta$ timed before the other ends being contacted.
\end{proof}

\begin{theorem}
Algorithm~\ref{alg:gossip} implements a probabilistic atomic object. 
\end{theorem}
\begin{proof}
The proof shows that it exists an ordering $\prec$ defined by the tags verifying Definition~\ref{def:patomicity}. 
This ordering is such that
$\pi_i \prec \pi_j$ is equivalent to either $\ms{tag}(\pi_i)=\ms{tag}(\pi_j)$ and 
$\pi_i$ is a write and $\pi_j$ is a read, or $\ms{tag}(\pi_i)<\ms{tag}(\pi_j)$.
Each property of Definition~\ref{def:patomicity} is proved separately.
\begin{enumerate}
\item Property 1 is deduced straightforwardly from the other Properties.
\item 
The proof is done in two parts.  First, we show that Property 2 holds if consultation phase 
of operation $\pi_2$ obtains an up-to-date value. Second, we show that this consultation 
phase obtains an up-to-date value with high probability. 
\begin{enumerate}
\item On the one hand, we denote by $\phi_i$ and by $\rho_i$ the respective consultation phase and propagation phase 
of any operation $\pi_i$. 
We show by contradiction that Property 1 holds if $\phi_2$ consults an up-to-date value.
By absurd, assume that it is false.  That is, assume that $\phi_2$ consults
an up-to-date value, the response of $\pi_1$ precedes the invocation
of $\pi_2$, and $\pi_2 \prec \pi_1$.
Since $\phi_2$ consults an up-to-date value, we have $\ms{tag}(\phi_2) \geq \ms{tag}(\pi_1)$.
Now there are two cases to consider: either $\pi_2$ is a read or a write.
First, if $\pi_2$ is a write then $\ms{tag}(\pi_2) > \ms{tag}(\phi_2) \geq \ms{tag}(\pi_1)$ 
by examination of the code of Algorithm~\ref{alg:gossip} (cf. Lines~\ref{line:dissem-tag-w}).
By definition of $\prec$, if $\ms{tag}(\pi_2) > \ms{tag}(\pi_1)$ and $\pi_2$ is a write, then
it can not happen that $\pi_2 \prec \pi_1$.
Second, if $\pi_2$ is a read then $\ms{tag}(\pi_2) = \ms{tag}(\phi_2) \geq \ms{tag}(\pi_1)$ 
by examination of the code of Algorithm~\ref{alg:gossip} (cf. Lines~\ref{line:dissem-tag-r}).
By definition of $\prec$, if $\ms{tag}(\pi_2) \geq \ms{tag}(\pi_1)$ and $\pi_2$ is a read, then
it can not happen that $\pi_2 \prec \pi_1$.
As a result, this contradicts the assumption, showing that Property 1 holds if $\phi_2$
obtains an up-to-date value.

\item On the other hand, Corollary~\ref{cor:prob-cons} shows 
that any consultation obtains the most up-to-date value 
with high probability. Since Property 2 holds if a consultation of $\pi_2$ consults
an up-to-date value, and since any consultation consults an up-to-date value with high probability, 
the result follows.
\end{enumerate}
\item Property 3 follows simply from the way tags are chosen.  Let $\pi_1$ and $\pi_2$ be any 
two operations.   On the one hand, if 
$\pi_1$ and $\pi_2$ are initiated at node $i$, then they have distinct tag counters.
On the other hand, if $\pi_1$ and $\pi_2$ are
initiated at two distinct nodes, then they have distinct tag identifiers
$i$ and $j$. As a result, two operations have different tags and 
either $\ms{tag}(\rho_1) > \ms{tag}(\rho_2)$ or 
$\ms{tag}(\rho_1) < \ms{tag}(\rho_2)$ holds.

\item Property 4 fails only if the read operation is unsuccessful.
The probability $P_\pi$ for an operation $\pi$ to be unsuccessful is lower than the 
probability $P_\phi$ that its consultation $\phi$ is unsuccessful. Since we know by Corollary~\ref{cor:prob-cons} 
that this later probability $P_\phi$ is very low ($P_\phi = e^{-\beta^2}$), the probability $P_\pi$ that an operation is 
unsuccessful is very low too ($P_\pi < e^{-\beta^2}$).
It follows that
Property 4 holds with high probability ($\geq 1-e^{-\beta^2}$).
\end{enumerate}
\end{proof}

\end{document}